\documentclass[twoside]{article}

\usepackage{ltexpprt}
\usepackage[margin = 3cm]{geometry}
\usepackage{float}
\usepackage{color}
\usepackage{graphicx}
\usepackage{amsmath}
\usepackage{amssymb}
\usepackage{times}
\usepackage{bbm}
\usepackage{tikz}
\usetikzlibrary{decorations.pathreplacing}
\usetikzlibrary{shapes}
\usepackage[colorlinks=true, linkcolor=blue, citecolor=red]{hyperref}

\let\originalleft\left
\let\originalright\right
\renewcommand{\left}{\mathopen{}\mathclose\bgroup\originalleft}
\renewcommand{\right}{\aftergroup\egroup\originalright}
\newcommand{\qed}{\qquad\vbox{\hrule height0.6pt\hbox{%
   \vrule height1.3ex width0.6pt\hskip0.8ex
   \vrule width0.6pt}\hrule height0.6pt
  }}

\newcommand{\EX}[1]{\mathbb{E}\left[#1\right]}
\renewcommand{\Pr}[1]{\mathbb{P}\left[#1\right]}

\definecolor{data_color}{RGB}{121, 222, 242}
\definecolor{alg_color}{RGB}{255,51,76}
\definecolor{thm_color}{rgb}{1.0, 0.6, 0.4}
\definecolor{red}{rgb}{1.0, 0, 0}
\definecolor{gray}{rgb}{0.5,0.5,0.5}

\newcommand{\support}[1]{\textup{supp}\left(#1\right)}
\renewcommand{\vec}[1]{\mathbf{#1}}

\newcommand{\e}{\epsilon}

\newcommand{\mc}{\mathsf{MC}}
\newcommand{\mm}[1]{\mathsf{MM}\left(#1\right)}

\newcommand{\MatchingSkeleton}{\mathsf{MatchingSkeleton}}
\newcommand{\block}{(P_i,\, Q_i, \, \alpha_i)}
\newcommand{\bigraph}{G = (P, \, Q, \, E)}

\newtheorem{definition}{Definition}

\newtheorem{remark}{Remark}
\newenvironment{proofof}[1]{\vspace{\baselineskip}\noindent{\bf Proof of #1:}}{$\qed$\par}

\begin{document}

\title{\Large Communication Efficient Coresets for Maximum Matching}
\author{Michael Kapralov\\EPFL\footnote{E-mail addresses: first.last@epfl.ch}
\and Gilbert Maystre\\EPFL
\and Jakab Tardos\\EPFL}

\date{\today}

\maketitle

\pagestyle{fancy}
\fancyhf{}
\fancyfoot[LE]{\hspace{1.75cm} \hfill \thepage \hfill}
\fancyfoot[RO]{\hfill \thepage \hfill \hspace{1.75cm}}

\setlength{\parindent}{0pt}

\begin{abstract} \small\baselineskip=9pt In this paper we revisit the problem of constructing {\em randomized composable coresets} for bipartite matching.  In this problem the input graph is randomly partitioned across $k$ players, each of which sends a single message to  a coordinator, who then must output a good approximation to the maximum matching in the input graph. Assadi and Khanna \cite{DBLP:conf/spaa/AssadiK17} gave the first such coreset, achieving a $1/9$-approximation by having every player send a maximum matching, i.e. at most $n/2$ words per player. The approximation factor was improved to $1/3$ by Bernstein et al. \cite{DBLP:conf/soda/AssadiBBMS19}.\\

In this paper, we show that the matching skeleton construction of Goel, Kapralov and Khanna \cite{DBLP:conf/soda/GoelKK12},  which is a carefully chosen (fractional) matching, is a randomized composable coreset that achieves a $1/2-o(1)$ approximation  using at most $n-1$ words of communication per player. We also show an upper bound of $2/3+o(1)$ on the approximation ratio achieved by this coreset.\end{abstract}

\section{Introduction}

Composable coresets is a generic technique for the analysis of large data, that has been shown to be effective for a variety of problems. In the context of graphs, the idea is to partition the edges of the input graph into $k$ parts, extract some small yet informative summary of each part, and recombine these summaries into a single graph without losing too much in the quality of the solution (the formal definition is presented in Section~\ref{section:coresets}). The small summary of each part is called the composable coreset. This versatile technique translates simply to algorithms in both the MPC and the randomized streaming models (Section 1.1 in \cite{DBLP:conf/spaa/AssadiK17}).\\

The study of {\it randomized} composable coresets in the context of approximating maximum matching was initiated by~\cite{DBLP:conf/spaa/AssadiK17} as the usefulness of {\it deterministic} composable coresets, where the initial partition of the input is arbitrary, was shown to be limited (see e.g. \cite{DBLP:conf/soda/AssadiKLY16,DBLP:conf/esa/Konrad15}). They proved that a maximum matching coreset, which contains $n/2$ edges, achieves nearly $1/9$ approximation, which was improved to nearly $1/3$ by~\cite{DBLP:conf/soda/AssadiBBMS19}. This paper further showed that the approximation quality of the maximum matching coreset is at best $1/2$, and proposed an alternative: the EDCS coreset. EDCS's achieve a nearly $2/3$ approximation as randomized composable coresets for maximum matching; they are however significantly denser, with size $n\cdot\text{poly}(\epsilon^{-1})\cdot\text{poly}(\log n)$ to achieve a $2/3-\epsilon$ approximation. More recently, the work of ~\cite{DBLP:conf/icml/AssadiBM19} gave a coreset of linear size in $n$ that achieves an approximation ratio of $1/2-\epsilon$ for small $\epsilon>0$, but at the expense of duplicating every edge $\Omega(1/\epsilon)$ times, increasing the communication accordingly. This is again prohibitively expensive for small $\epsilon$.\\

As the main result of this paper, we propose a small composable coreset of size at most $n-1$, which nonetheless achieves a $1/2$ approximation ratio, without the need for the duplication of edges.

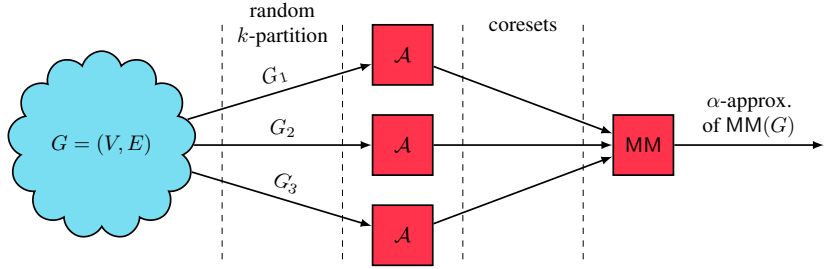
\begin{figure}[ht]
\vskip 0.2in
\begin{center}
\scalebox{0.8}{
\begin{tikzpicture}
\node[cloud, cloud puffs=13, draw, thick, fill=data_color, minimum width=3cm, minimum height=3cm, align=center] (data) at (-2, 0) {$G = (V,E)$};
\node[rectangle, draw, thick, fill=alg_color, minimum width=1cm, minimum height=1cm] (a1) at (3, 1.5) {$\mathcal{A}$};
\node[rectangle, draw, thick, fill=alg_color, minimum width=1cm, minimum height=1cm] (a2) at (3, 0) {$\mathcal{A}$};
\node[rectangle, draw, thick, fill=alg_color, minimum width=1cm, minimum height=1cm] (a3) at (3, -1.5) {$\mathcal{A}$};

\node[rectangle, draw, thick, fill=alg_color, minimum width=1cm, minimum height=1cm] (mm) at (7, 0) {$\mathsf{MM}$};

\path[draw, -latex, thick] (data) -- (a1) node [midway, above, sloped] {$G_{1}$};
\path[draw, -latex, thick] (data) -- (a2) node [midway, above, sloped] {$G_{2}$};
\path[draw, -latex, thick] (data) -- (a3) node [midway, above, sloped] {$G_{3}$};
\path[draw, -latex, thick] (a1) -- (mm);
\path[draw, -latex, thick] (a2) -- (mm);
\path[draw, -latex, thick] (a3) -- (mm);
\path[draw, -latex, thick] (mm) -- (10, 0) node [midway, above, sloped, align=center] {$\alpha$-approx.\\ of $\mm{G}$};

\path[draw, dashed] (0, 2) -- (0, -2);
\path[draw, dashed] (2, 2) -- (2, -2);
\node[rectangle, align=center, minimum width=2, minimum height=1] at (1, 2) {random\\$k$-partition};

\path[draw, dashed] (4, 2) -- (4, -2);
\path[draw, dashed] (6, 2) -- (6, -2);
\node[rectangle, align=center, minimum width=2, minimum height=1] at (5, 2) {coresets};
\end{tikzpicture}}
\caption{A visual representation of randomized composable coresets. $G$ is first partitioned into $k$ parts randomly. Then each of those parts is reduced into coresets independently using an algorithm $\mathcal{A}$. A maximum matching  is then computed amongst the recombination of all the coresets.}
\label{figure:coresetconcept}
\end{center}
\vskip -0.2in
\end{figure}

\begin{theorem}
There exists a $1/2-o(1)$-approximate randomized composable coreset with size $n-1$ for bipartite maximum matching as long as $k=\omega(1)$, $\mm{G}=\omega(k\log n)$.
\end{theorem}

\subsection*{Intuition behind our construction}
Our coreset is inspired by the previous best known `small' randomized composable coreset, the maximum matching coreset. The main challenge in analyzing the performance of picking any maximum matching as a coreset lies in the fact that graphs that do not admit a perfect matching generally admits many different maximum matchings. To circumvent this, we propose to use \textit{any matching skeleton} (a structure introduced by~\cite{DBLP:conf/soda/GoelKK12}, and later rediscovered by~\cite{DBLP:conf/soda/BernsteinHR18}) as a coreset. This is essentially a carefully chosen `canonical' fractional matching that matches vertices as uniformly as possible (see Section~\ref{section:matchingskeletoncoreset}). Such a fractional matching can always be selected to be supported on a forest by a simple argument similar to the one that establishes the integrality of the bipartite matching polytope, meaning that the support size is never larger than $n-1$. The fact that the coreset is essentially an `optimally spread out' maximum matching leads to a rather simple proof of the approximation ratio of $1/2-o(1)$, which we present in Section~\ref{section:main}. In Section~\ref{section:limitations}, we show that any matching skeleton does not provide a better than $2/3$ approximation to maximum matching, leaving some amount of room for improvement of our approximation ratio bound.

\subsection*{Previous results}
Coresets have been studied in a variety of contexts \cite{DBLP:conf/spaa/AssadiK17,DBLP:conf/kdd/BadanidiyuruMKK14,DBLP:conf/nips/BalcanEL13,DBLP:conf/nips/BateniBLM14,DBLP:conf/pods/IndykMMM14,DBLP:conf/stoc/MirrokniZ15,DBLP:conf/nips/MirzasoleimanKSK13} (also see e.g. \cite{DBLP:conf/soda/AhnGM12,DBLP:conf/pods/AhnGM12,DBLP:conf/soda/AssadiKLY16,DBLP:conf/stoc/BhattacharyaHNT15,DBLP:journals/algorithmica/BulteauFKP16,DBLP:conf/soda/ChitnisCEHMMV16,DBLP:conf/focs/KapralovLMMS14,DBLP:conf/podc/KapralovW14,DBLP:conf/mfcs/McGregorTVV15, DBLP:journals/corr/abs-2007-14204} for the related work in the linear sketching context). Related to our problem, maximum matching approximation has been widely studied in low space regimes such as MPC \cite{DBLP:conf/soda/AssadiBBMS19,DBLP:conf/spaa/HarveyLL18,DBLP:conf/stoc/CzumajLMMOS18,DBLP:conf/spaa/LattanziMSV11,DBLP:conf/icml/AssadiBM19} and streaming \cite{DBLP:conf/approx/McGregor05,DBLP:conf/soda/GoelKK12,DBLP:conf/approx/KonradMM12,DBLP:journals/talg/PazS19,DBLP:conf/spaa/AssadiK17,DBLP:conf/soda/AssadiB19}. In particular \cite{DBLP:conf/soda/AssadiBBMS19} achieves nearly $2/3$ approximate maximum matching in two MPC rounds and $\widetilde O(\sqrt{mn}+n)$ space per machine, using randomized composable coresets of size $O(n\log n)$.

\section{Randomized Composable Coresets}
\label{section:coresets}

\begin{definition}
Let $G = (V, E)$ be a graph and $k \in \mathbb{N}$ and integer. A {\bf random $k$-partition} of $G$ is a set of $k$ random subgraphs $\{G_i=(V,\, E_i)\}_{i \in [k]}$ of $G$, where each edge $e \in E$ is sent uniformly at random to exactly one of the $E_i$.
\end{definition}

\begin{definition}\cite{DBLP:conf/stoc/MirrokniZ15}
\label{definition:coreset}
Let $\mathcal{A}$ be an algorithm that takes as input a graph $H$ and returns a subgraph $\mathcal{A}(H) \subseteq H$. We say that $\mathcal{A}$ outputs an {\bf $\alpha$-approximate randomized composable coreset} for the maximum matching problem if given any graph $G = (V, E)$, any $k \in \mathbb{N}$ and a random $k$-partition of $G$ we have
$$\alpha \cdot \mm{G} \leq \EX{\mm{\mathcal{A}\left(G_1 \right) \cup \dots \cup \mathcal{A}\left(G_k \right)}}$$
where the expectation is taken over the randomness of the partition. The {\bf size} of the coreset is the number of edges returned by $\mathcal{A}$.
\end{definition} 

\begin{remark}
Throughout this paper we will assume some natural bounds on the parameter $k$. Firstly, similarly to \cite{DBLP:conf/spaa/AssadiK17, DBLP:conf/soda/AssadiBBMS19}, we suppose that the maximum matching size of the input graph $\mm{G}=\omega(k\log n)$. This allows us to argue concentration at various places in the analysis, and is a natural assumption: The regime where $\mm{G}$ is smaller is handled in \cite{DBLP:journals/corr/ChitnisCEHMMV15}. We will further make the natural assumption that $k=\omega(1)$, that is we parallelize over a superconstant number of machines.
\end{remark}

\section{Preliminaries and Notation}
Throughout the paper we consider bipartite graphs, denoted by $G = (P,\, Q,\, E)$, where the vertex-sets $P$ and $Q$ are the two sides of the bipartition, and $E$ is the edge-set. We let $n=|P\cup Q|$ denote the number of vertices in $G$ and $m=|E|$ denote the number of edges. For a vertex $v\in P\cup Q$ of $G$ we write $\Gamma_G(v)$ to denote the set of neighbors of $v$ in $G$, or $\Gamma(v)$ if $G$ is clear from context. Similarly, for a set $S\subseteq P\cup Q$ we write $\Gamma_G(S)$ or $\Gamma(S)$ to the denote the neighborhood of the set in $G$.

\begin{definition}
A {\bf matching} in a graph is a set of edges such that no two of them share an end point. The {\bf maximum matching size} of a graph is the maximum possible size of a matching in it; we usually denote it $\mm{G}$.
\end{definition}

\begin{definition}
Given a graph $G = (P, Q, E)$, a {\bf fractional matching} is a set of non-negative edge weights $\vec{x}: E \to [0,1]$ such that no vertex has more than unit weight adjacent on it:
$$\forall v\in P\cup Q:\sum_{w\in\Gamma(v)}x_{vw}\le1$$
The {\bf size} of a fractional matching is the sum of all edge-weights.
\end{definition}

Note that an integral fractional matching corresponds to the classical matching definition. We will also use the extended notion of $\alpha$-matching of \cite{DBLP:conf/soda/GoelKK12}, which are classical fractional matching with a changed constraint for one side of the bipartition.

\begin{definition}
Given a graph $G = (P, Q, E)$, a {\bf $\alpha$-matching with respect to $P$} is a set of non-negative edge weights $\vec{x}: E \to [0, 1]$ that saturates each vertex of $P$ fractionally exactly $\alpha$ times and each vertex of $Q$ at most once.
\end{definition}

\begin{definition}
A vertex cover is a set of vertices $\Phi\subseteq P\cup Q$ such that all edges have at least one end point in $\Phi$.
\end{definition}

The following theorem is a fundamental fact about bipartite graphs, on which we will be relying throughout the paper.

\begin{theorem}\label{thm:vc=mm}
For any bipartite graph, the size of the maximum matching, the size of the maximum fractional matching, and the size of the minimum vertex cover are equal.
\end{theorem}

\begin{corollary}
\label{corollary:m-vc}
If a matching and a vertex cover have the same size, both are optimal.
\end{corollary}

Furthermore, we will rely on the following concentration inequality.

\begin{theorem}[Chernoff bound, see \cite{DBLP:books/daglib/0021015}]
Let $Y = \sum_{i = 1}^n X_i$ be the sum of $n$ independent binary random variable each with $\Pr{X_i = 1} = p_i$. Let $\mu_Y = \EX{Y} = \sum_{i = 1}^n p_i$. Then, for any $\e \in (0,1)$, we have:
$$\Pr{X \notin (1 \pm \epsilon) \mu_Y} \leq 2e^{-\frac{\epsilon^2 \mu_Y}{3}}$$ 
\label{theorem:chernoff}
\end{theorem}

\section{Our coreset: the matching skeleton}
\label{section:matchingskeletoncoreset}

In this section, we recall the notion of matching skeleton, introduced by \cite{DBLP:conf/soda/GoelKK12} and later rediscovered by \cite{DBLP:conf/soda/BernsteinHR18}. We simplify slightly the original definitions and results to suit our needs. We also introduce a new related object, the canonical vertex cover which is central to our proof.\\

We define a partition of the vertex set of $G$ into subgraphs of varying vertex expansion as follows. For each $i = 1,\, \dots$ we define a tuple $\block$ iteratively as follows, starting with $G_0 = G$, $P_0=P$:

\begin{enumerate}
\item Let $\alpha_i = \min_{\emptyset\neq S\subseteq P_{i-1}} \frac{|\Gamma_{G_{i-1}}(S)|}{|S|}$
\item Let $P_i=$ largest $S\subseteq P_{i-1}$ such that $\frac{|\Gamma_{G_{i-1}}(S)|}{|S|}=\alpha_i$
\item Let $Q_i= \Gamma_{G_{i-1}}(P_i)$
\item $G_i = G_{i-1} \setminus (P_i \cup Q_i)$
\end{enumerate}

This process continues until $G_i$ is empty.

\begin{definition}
We call each $\block$ a {\bf block} and $\alpha_i$ its {\bf expansion level}, which is carried over to the vertices of the block using the notation $\alpha(v) := \alpha_i$ for $v \in P_i \cup Q_i$. We call the collection $\{\block\}_{i\in[k]}$ the {\bf block decomposition} of $G$.
\end{definition}

\begin{remark}
A practical way to find $\alpha_1$ is to solve several max-flow instances. For some $\alpha \in \mathbb{R}_+$, let $G_\alpha$ be a copy of $G$ where edges are directed from $P$ to $Q$ with an infinite weight. Also part of $G_\alpha$ is a source vertex $s$ which has edges directed toward each $p \in P$ with weight $\alpha$ and a sink vertex $t$ with unit-weight edges incoming from each $q \in Q$. Observe that $\alpha_1 = \inf\{\alpha \in \mathbb{R}_+: |\mc(\alpha)| > 1\}$ where $\mc(\alpha)$ is the min-cut containing $s$ in $G_\alpha$. Finding $\alpha_1$ thus reduces to solving max-flow instances with increasing $\alpha$ until a non-trivial min-cut is found. This cut actually consists of $P_1$ and $Q_1$ together with $s$. The remaining of the partition is obtained by repeating this argument.
\end{remark}

We now recall the main properties of the block decomposition of $G$, most of which comes from Section 3 of \cite{DBLP:conf/soda/GoelKK12}.

\begin{lemma}[\cite{DBLP:conf/soda/GoelKK12}]
\label{lemma:blockproperties}
Let $\{\block\}_{i \in [k]}$ be the block partition of $G$. The sequence $(\alpha_i)_{i \in [k]}$ is strictly increasing and such that $\alpha_i = |Q_i|/|P_i|$. Also, for any $i \in [k]$:
$$\Gamma(P_i) \subseteq \bigcup_{j \leq i} Q_j$$
\end{lemma}

Intuitively, each block $P_i\cup Q_i$ is associated with a certain expansion of the $P_i$ side, namely $\alpha_i$. The expansion of the block cannot be greater than $\alpha_i$, as $|Q_i|=\alpha_i|P_i|$. However, it is also no less than $\alpha_i$, as the entire block admits of an $\alpha_i$-matching with respect to $P_i$.

\begin{lemma}[\cite{DBLP:conf/soda/GoelKK12},\cite{DBLP:conf/soda/BernsteinHR18}]
Let $\bigraph$ be a graph together with its block decomposition $\{\block\}_{i\in[k]}$. For each $i\in[k]$ there is an $\alpha_i$-matching of $P_i\cup Q_i$ with respect to $P_i$.
\end{lemma}
\begin{remark}[\cite{DBLP:conf/soda/GoelKK12}]
The above $\alpha$-matchings can easily be made to have cycle-free supports, by eliminating cycles through standard techniques.
\end{remark}

Now that the block decompositon of a graph is introduced, we can define matching skeletons which are simply the union of the above introduced cycle-free $\alpha$-matching for each block.

\begin{definition}[Matching skeleton \cite{DBLP:conf/soda/GoelKK12}]
Let $\bigraph$ be a graph together with its block decomposition $\{\block\}_{i \in [k]}$. For each $i \in [k]$, let $\vec{x}_i: (P_i \times Q_i) \cap E \to [0, 1]$ be a cycle-free $\alpha_i$-matching. We call 
$$H = \bigcup_{i \in [k]} \support{\vec{x}_i}$$
a matching skeleton of $G$. See Figure \ref{figure:coreset} for a visual example.
\end{definition}

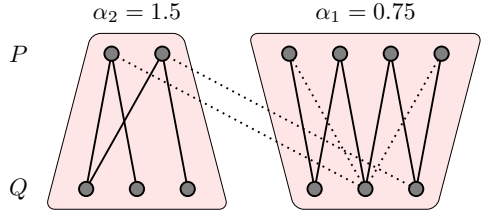
\begin{figure}[ht]
\vskip 0.2in
\begin{center}
\scalebox{0.9}{
\begin{tikzpicture}

\coordinate (bca) at (6, 0);
\coordinate (bcb) at (6.75, 0);
\coordinate (bcc) at (7.5, 0);
\coordinate (bsa) at (6.375, 2);
\coordinate (bsb) at (7.125, 2);

\coordinate (cca) at (9.375, 0);
\coordinate (ccb) at (10.125, 0);
\coordinate (ccc) at (10.875, 0);
\coordinate (csa) at (9, 2);
\coordinate (csb) at (9.75, 2);
\coordinate (csc) at (10.5, 2);
\coordinate (csd) at (11.25, 2);

\draw[fill=red!10, rounded corners] (5.4, -0.3) -- (8.1, -0.3) -- (7.425, 2.3) -- (6.075, 2.3)-- cycle;
\draw[fill=red!10, rounded corners] (9.075, -0.3) -- (11.175, -0.3) -- (11.85, 2.3) -- (8.4, 2.3)-- cycle;

\path[draw=black, thick] (bca) -- (bsa);
\path[draw=black, thick] (bca) -- (bsb);
\path[draw=black, thick] (bcb) -- (bsa);
\path[draw=black, thick] (bcc) -- (bsb);

\path[draw=black, thick] (cca) -- (csa);
\path[draw=black, thick] (cca) -- (csb);
\path[draw=black, thick] (ccb) -- (csb);
\path[draw=black, thick] (ccb) -- (csc);
\path[draw=black, thick] (ccc) -- (csc);
\path[draw=black, thick] (ccc) -- (csd);
\path[draw=black, dotted, thick] (ccb) -- (csa);
\path[draw=black, dotted, thick] (ccb) -- (csd);

\path[draw=black, thick, dotted] (ccb) -- (bsa);
\path[draw=black, thick, dotted] (ccc) -- (bsb);

\node[circle, thick, draw, fill=gray, minimum size=6, inner sep=0pt, outer sep=0pt] at (bca) {};
\node[circle, thick, draw, fill=gray, minimum size=6, inner sep=0pt, outer sep=0pt] at (bcb) {};
\node[circle, thick, draw, fill=gray, minimum size=6, inner sep=0pt, outer sep=0pt] at (bcc) {};
\node[circle, thick, draw, fill=gray, minimum size=6, inner sep=0pt, outer sep=0pt] at (bsa) {};
\node[circle, thick, draw, fill=gray, minimum size=6, inner sep=0pt, outer sep=0pt] at (bsb) {};

\node[circle, thick, draw, fill=gray, minimum size=6, inner sep=0pt, outer sep=0pt] at (cca) {};
\node[circle, thick, draw, fill=gray, minimum size=6, inner sep=0pt, outer sep=0pt] at (ccb) {};
\node[circle, thick, draw, fill=gray, minimum size=6, inner sep=0pt, outer sep=0pt] at (ccc) {};
\node[circle, thick, draw, fill=gray, minimum size=6, inner sep=0pt, outer sep=0pt] at (csa) {};
\node[circle, thick, draw, fill=gray, minimum size=6, inner sep=0pt, outer sep=0pt] at (csb) {};
\node[circle, thick, draw, fill=gray, minimum size=6, inner sep=0pt, outer sep=0pt] at (csc) {};
\node[circle, thick, draw, fill=gray, minimum size=6, inner sep=0pt, outer sep=0pt] at (csd) {};

\node (t2) at (6.75, 2.6) {$\alpha_2 = 1.5$};
\node (t3) at (10.125, 2.6) {$\alpha_1 = 0.75$};

\node at (5, 0) {$Q$};
\node at (5, 2) {$P$};
\end{tikzpicture}
}
\caption{A graph $\bigraph$ and its block partition. A cycle-free matching skeleton is shown with solid edges. Due to the construction of the block partition, edges between $P_2$ and $Q_1$ cannot be part of any matching skeleton. Also, from Lemma \ref{lemma:blockproperties}, no edge can exist between $Q_2$ and $P_1$.}
\label{figure:coreset}
\end{center}
\vskip -0.2in
\end{figure}

\begin{remark}
The matching skeleton coreset has size at most $|P| + |Q| -1$, as it is always a forest.
\end{remark}

We now describe a special kind of vertex cover, related with the block partition of a graph, which will be a crucial tool in analyzing the quality of our coreset.

\begin{definition}[Canonical vertex cover]
\label{definition:canonicalvc}
Given a graph $\bigraph$ together with its block decomposition, we call
$$\Phi=\{q\in Q \,|\, \alpha(q) < 1\} \cup \{p\in P \,|\, \alpha(p) \geq 1\}$$
the {\bf canonical vertex cover} of $G$. That is, in each block we take the smaller side of the bipartition.
\end{definition}

\begin{lemma}
\label{lemma:canonicalvc}
The canonical vertex cover is a minimum vertex cover.
\end{lemma}
\begin{proof}
First we show that $\Phi$ is indeed a vertex cover. Suppose there exists an edge $\{p,\, q\}$ not adjacent on $\Phi$ and let $p \in P_i$ and $q \in Q_j$. By definition of the canonical vertex cover, this means that $\alpha_j \geq 1 > \alpha_i$ and hence $i < j$ using monotonicity of the expansion levels (Lemma \ref{lemma:blockproperties}). In turn, this implies that $\Gamma(p) \not\subseteq \bigcup_{\ell \leq i} Q_\ell$: a contradiction with Lemma \ref{lemma:blockproperties}.\\

We now proceed to show that $\Phi$ is minimum, by showing that there exists a fractional matching of size $|\Phi|$ (See Corollary~\ref{corollary:m-vc}). We define this fractional matching block-by-block: Consider the block $\block$.
\begin{itemize}
\item If $\alpha_i < 1$, then $\Phi\cap (P_i \cup Q_i)$ is exactly $Q_i$. In this case, we can take the $\alpha_i$-matching with respect to $P_i$ as our fractional matching. This will have size $\alpha_i|P_i|=|Q_i|$, exactly as desired.
\item On the other hand, if $\alpha_i \geq 1$, then $\Phi\cap (P_i \cup Q_i)$ is exactly $P_i$. Thus, an $\alpha_i$-matching with respect to $P_i$ \textit{scaled down by a factor of $\alpha_i$} is a valid fractional matching of the block, and has size $|P_i|$.

\end{itemize}
\end{proof}

The above deduction also shows that any matching skeleton contains a maximum matching. Therefore, the matching skeleton coreset performs at least as well as the maximum matching coreset of \cite{DBLP:conf/spaa/AssadiK17}. In particular, this directly yields a lower bound of $1/3$ on the approximation ratio of our coreset. However, a matching skeleton retains more information from the input graph, as the entire block partition can be recovered from it. This allows for a better approximation ratio as Section \ref{section:main} demonstrates.

\begin{remark}
Let us draw the parallels between the server flows of \cite{DBLP:conf/soda/BernsteinHR18} and the notion of matching skeleton of \cite{DBLP:conf/soda/GoelKK12}. In the context of \cite{DBLP:conf/soda/BernsteinHR18}, the support of a realization of the balanced server flow is simply a matching skeleton.  The balancedness condition corresponds to the neighboring property of Lemma \ref{lemma:blockproperties}. Finally, the server flow values of \cite{DBLP:conf/soda/BernsteinHR18} are exactly the expansion levels.
\end{remark}

Finally, we prove a structural result about the robustness of the block decomposition under changes to the edge-set of $G$. This will be crucial to our proofs in both Sections~\ref{section:main} and~\ref{section:limitations}.

\begin{lemma}\label{lemma:robust}
Let $\bigraph$ be a graph together with its block decomposition $\{\block\}_{i\in[k]}$ and $H$ a matching skeleton of $G$. Let $E^+, E^- \subseteq P \times Q$ be two sets of edges such that $E^- \cap H = \emptyset$ and for any $\{p,\, q\} \in E^+$, $\alpha(p) \geq \alpha(q)$.\\

Denote by $G' = (P, Q, E')$ the modification of $G$ with edge set $E' = (E \cup E^+) \setminus E^-$. The block decomposition of $G'$ is still $\{\block\}_{i\in[k]}$, and therefore $H$ remains a valid matching skeleton of $G'$.
\end{lemma}

\begin{proof}
We will use $\Gamma(S)$, $\Gamma'(S)$ and $\Gamma_H(S)$ to denote the neighborhood of some set $S$ in the graphs $G$, $G'$ and $H$ respectively. Consider now the first step of the block decomposition of $G'$. We first prove that no set $S\subseteq P$ has lower expansion than $\alpha_1$ in $G'$. Consider any set $S\subseteq P$. We can lower bound the size of $\Gamma'(S)$ by $\Gamma_H(S)$ since $H\subseteq G'$. Moreover, we note that $H$ contains the support of an $\alpha_i$-matching with respect to $P_i$, in block $(P_i\cup Q_i)$, for each $i$. Therefore, the expansion of any subset in $P_i$ is at least $\alpha_i$ and
$$|\Gamma_H(S\cap P_i) \cap Q_i| = |\Gamma_H(S\cap P_i)|\ge\alpha_i|S\cap P_i|.$$
The equality comes from the fact that a matching skeleton contains no edge crossing two blocks. Using this, we have:
$$|\Gamma'(S)| \ge|\Gamma_H(S)| \ge\sum_{i=1}^k|\Gamma_H(S\cap P_i)\cap Q_i| =\sum_{i=1}^k|\Gamma_H(S\cap P_i)| \ge\sum_{i=1}^k\alpha_i|S\cap P_i| \ge\alpha_1|S|$$
Note that the statement is true with strict inequality when $S\not\subseteq P_1$. On the other hand, the expansion of $P_1$ in $G'$ is exactly $\alpha_1$, as $\Gamma'(P_1)=Q_1$. This is because $E^+$ cannot have any edge between $P_1$ and $Q\setminus Q_1$.\\
	
We thus have proven that the first block in the decomposition of $G'$ is $(P_1, \, Q_1, \, \alpha_1)$. One can then proceed by induction on $i$ to prove that the same is true for the $i^\text{th}$ block. The argument is identical to the base case by observing that since $E^+$ cannot have edges between $P_i$ and $\bigcup_{j=i+1}^kQ_j$, it does not increase the expansion of $P_i$.
\end{proof}

\section{Main Result}
\label{section:main}

Having defined the matching skeleton coreset, we now prove a lower bound of nearly $1/2$ on its effectiveness. This improves upon any known lower bound for a randomized composable coreset of size $O(n)$ for the maximum matching problem.
\begin{theorem}
$\MatchingSkeleton(G)$ constitutes a $\left(1/2-o(1)\right)$-approximate randomized composable coreset for maximum matching in any bipartite graph $\bigraph$ where $k=\omega(1)$ and the maximum matching size $\mm{G}=\omega(k\log n)$.
\label{theorem:halfapprox}
\end{theorem}
\begin{proof}
Our analysis is inspired by the maximum matching coreset analysis of ~\cite{DBLP:conf/soda/AssadiBBMS19}, however, we achieve a better approximation ratio using more subtle techniques. Let $\mu$ denote $\mm{G}$. Recall that by the definition of randomized composable coresets (Definition~\ref{definition:coreset}) we must randomly edge-partition $\bigraph$ into $k$ subgraphs $G_1,\ldots,G_k$, and show that the union of each coresets,
\begin{equation}
\widetilde G=\bigcup_{i=1}^k\MatchingSkeleton(G_i),\label{equation:main}
\end{equation}
has an expected maximum matching size of $\mu\cdot(1/2-o(1))$, over the randomness of the $k$-partition. We begin by choosing an arbitrary maximum matching $M^*$ of $G$. We separate $G$ in two parts: $M^*$ and $G^-:=G\backslash M^*$ for the purposes of analysis, and for every $i=1,\ldots, k,$ let $G^-_i:=G_i\cap G^-$.\\

We will show the stronger statement that even under {\bf adversarial partitioning} of $G^-$, Equation~\ref{equation:main} holds, as long as $M^*$ is partitioned randomly. From now on we will assume that the partition into $G^-_1,\ldots,G^-_k$ is fixed arbitrarily; we will show that either at least one of $G^-_i$ contains a large matching or $M^*\cap\widetilde G$ is large.\\

Consider an arbitrary $k$-partitioning of $G^-$ into $G^-_1,\ldots,G^-_k$ and let the maximum matching size of $G^-_i$ be $\mu^-_i$. If even one of $\mu^-_i$ is at least $\mu/2$, we are done. Indeed, following Lemma \ref{lemma:canonicalvc}, any matching skeleton of $G_i$ will contain a maximum matching, that is a matching of size $\mm{G_i}\ge\mu^-_i\ge\mu/2$, and hence so will $\widetilde G$. Therefore, we can focus on the case where $\max_{i=1}^k\mu^-_i\le\mu/2$ and use the following lemma, which is our main technical contribution:

\begin{lemma}[Main lemma]
\label{lemma:main}
Consider an arbitrary partitioning of $G^-$ where $\max_{i=1}^k\mu^-_i<\mu/2$. Let $e$ be a uniformly random element of $M^*$. Then
$$\Pr{e\in\widetilde G}\ge1/2-o(1),$$
where the probability is taken over the randomness of the partitioning of $M^*$ as well as the randomness of the choice of $e$.
\end{lemma}

The above lemma relies on a subtle probabilistic argument, and is formulated in terms of a uniformly random edge of $M^*$ for technical reasons. However, an immediate consequence of it is that at least nearly half of the edges of $M^*$ will be taken in $\widetilde G$. This follows by linearity of expectation:
$$\EX{\left\vert M^*\cap\widetilde G\right\vert} =\EX{\sum_{e\in M^*}\mathbbm1(e\in\widetilde G)} =\sum_{e\in M^*} \Pr{e\in\widetilde G} \ge\mu\cdot(1/2-o(1)).$$
where the last inequality follows by Lemma~\ref{lemma:main}. We have proven that Equation~\ref{equation:main} holds under adversarial partitioning of $G^-$ both when $\max_{i=1}^k\mu^-_i\ge\mu/2$ and when $\max_{i=1}^k\mu^-_i<\mu/2$, which implies the statement of the theorem.
\end{proof}
\vspace{0.1in}
We conclude the analysis of the $\MatchingSkeleton$ coreset by proving Lemma~\ref{lemma:main}.

\begin{proofof}{Lemma~\ref{lemma:main}}
Without loss of generality we may assume that $e\in G_1$. We know that the maximum matching size of $G^-_1$ is at most $\mu/2$. Consider now adding to $G^-_1$ all edges of $M^*\cap G_1$ {\it except for} $e$. Since the size of $M^*\cap G_1\backslash\{e\}$ is at most $2\mu/k$ with high probability by Theorem~\ref{theorem:chernoff}, the maximum matching size does not increase by more than $2\mu/k$.\\
    
We base our analysis on fixing the outcome of the random graph $G_1\backslash\{e\}$ to be some fixed $H$. We refer to this event that $G_1\backslash\{e\}=H$ as $\mathcal E(H)$. Suppose that indeed the maximum matching size of $H$ is at most $\mu\cdot(1/2+2/k)$, and hence that $H$ has a canonical vertex cover $\Phi$ of this size. Recall from Definition~\ref{definition:canonicalvc} that the canonical vertex cover contains exactly the vertices of $Q$ with $\alpha$-value {\it strictly less} than one and the vertices of $P$ with $\alpha$-values {\it at least} one. Therefore, any new edge added to $H$ that is {\it not} adjacent on $\Phi$ {\it must be} included in any matching skeleton, as we show in the following paragraph.\\
    
Indeed, consider $e=\{p,\, q\}$ to be such an edge, and suppose that there exists some matching skeleton $H$ of $G_1$ where $e$ is not included. This means, by Lemma~\ref{lemma:robust} with $E^-=\{e\}$ and $E^+=\emptyset$ that the block decompositions of $G_1$ and $H$ are identical. However, by definition of the canonical vertex cover $\Phi$ for $H$ and because $p,q \notin \Phi$, we have $\alpha_H(p) < 1 \leq \alpha_H(q)$. This implies that in the block partition of $G_1$, $p \in P_i$ and $q \in Q_j$ with $i < j$, which is a contradiction of Lemma~\ref{lemma:blockproperties}.\\

Consequently, if $e$ is not adjacent on $\Phi$, it must be taken into $\MatchingSkeleton(G_1)$. The last important observation is that the distribution of $e$, when conditioned on $\mathcal E(H)$, is uniform on $M^*\backslash H$. Indeed, this conditioning in no way breaks the symmetry between the unsampled edges $M^*\backslash H$, and $e$ is equally likely to be any of them. Therefore, $e$ is uniformly distributed among at least $\mu\cdot(1-2/k)$ edges among which at most $\mu\cdot(1/2+2/k)$ are adjacent on $\Phi$: Conditioned on $\mathcal E(H)$, where $\mm{H}\le\mu\cdot(1/2+2/k)$, the probability that $e$ is not adjacent on $\Phi$ and therefore $e\in\MatchingSkeleton(G_1)$ is at least $1/2-o(1)$.\\
    
The above deduction was made with the assumption that $\mm{H}\le\mu\cdot(1/2+2/k)$. However, recall that this happens with high probability by Theorem~\ref{theorem:chernoff}, therefore we can extend the result to full generality. Consider the possible outcomes of $G_1\backslash\{e\}$ to form the family $\mathcal H$. We can split $\mathcal H$ into the disjoint union of $\mathcal H_0$ and $\mathcal H^*$, where $\mathcal H^*$ comprises the anomalous outcomes where the maximum matching size of $H$ is greater than $\mu\cdot(1/2+2/k)$. Then,
\begin{align*}
\Pr{e\in \MatchingSkeleton(G_1)} &= \sum_{H\in\mathcal H} \Pr{\mathcal E(H)}\Pr{e\in\MatchingSkeleton(G_1)|\mathcal E(H)}\\
&\ge \sum_{H\in\mathcal H_0} \Pr{\mathcal E(H)}  \Pr{e\in\MatchingSkeleton(G_1)|\mathcal E(H)}\\
&\ge \sum_{H\in\mathcal H_0} \Pr{\mathcal E(H)}\cdot(1/2-o(1))\\
&=\Pr{G_1\backslash\{e\}\not\in\mathcal H^*}\cdot(1/2-o(1))\\
&\ge 1/2-o(1),
\end{align*}
as desired.
\end{proofof}

\begin{figure}[ht]
\vskip 0.2in
\begin{center}
\scalebox{0.9}{
\begin{tikzpicture}
\coordinate (aca) at (0, 0);
\coordinate (acb) at (0.75, 0);
\coordinate (acc) at (1.5, 0);
\coordinate (acd) at (2.25, 0);
\coordinate (ace) at (3, 0);
\coordinate (acf) at (3.75, 0);
\coordinate (acg) at (4.5, 0);
\coordinate (asa) at (1.125, 2);
\coordinate (asb) at (1.875, 2);
\coordinate (asc) at (2.625, 2);
\coordinate (asd) at (3.375, 2);

\coordinate (bca) at (6, 0);
\coordinate (bcb) at (6.75, 0);
\coordinate (bcc) at (7.5, 0);
\coordinate (bsa) at (6.375, 2);
\coordinate (bsb) at (7.125, 2);

\coordinate (cca) at (9.375, 0);
\coordinate (ccb) at (10.125, 0);
\coordinate (ccc) at (10.875, 0);
\coordinate (csa) at (9, 2);
\coordinate (csb) at (9.75, 2);
\coordinate (csc) at (10.5, 2);
\coordinate (csd) at (11.25, 2);

\draw[fill=red!10, rounded corners] (-0.6, -0.3) -- (5.1, -0.3) -- (3.675, 2.3) -- (0.825, 2.3)-- cycle;
\draw[fill=red!10, rounded corners] (5.4, -0.3) -- (8.1, -0.3) -- (7.425, 2.3) -- (6.075, 2.3)-- cycle;
\draw[fill=red!10, rounded corners] (9.075, -0.3) -- (11.175, -0.3) -- (11.85, 2.3) -- (8.4, 2.3)-- cycle;

\path[draw=black, thick] (acc) -- (csa);
\path[draw=black, thick] (ace) -- (csb);
\path[draw=black, thick] (bcb) -- (csc);
\path[draw=black, thick] (bcc) -- (csd);

\path[draw=black, thick, dotted] (aca) -- (asa);
\path[draw=black, thick, dotted] (bca) -- (asc);
\path[draw=black, thick, dotted] (cca) -- (asd);
\path[draw=black, thick, dotted] (ccb) -- (bsa);
\path[draw=black, thick, dotted] (ccc) -- (bsb);

\node[circle, thick, draw, fill=gray, minimum size=6, inner sep=0pt, outer sep=0pt] at (aca) {};
\node[circle, thick, draw, fill=gray, minimum size=6, inner sep=0pt, outer sep=0pt] at (acb) {};
\node[circle, thick, draw, fill=gray, minimum size=6, inner sep=0pt, outer sep=0pt] at (acc) {};
\node[circle, thick, draw, fill=gray, minimum size=6, inner sep=0pt, outer sep=0pt] at (acd) {};
\node[circle, thick, draw, fill=gray, minimum size=6, inner sep=0pt, outer sep=0pt] at (ace) {};
\node[circle, thick, draw, fill=gray, minimum size=6, inner sep=0pt, outer sep=0pt] at (acf) {};
\node[circle, thick, draw, fill=gray, minimum size=6, inner sep=0pt, outer sep=0pt] at (acg) {};
\node[circle, thick, draw, fill=gray, minimum size=6, inner sep=0pt, outer sep=0pt] at (asa) {};
\node[circle, thick, draw, fill=gray, minimum size=6, inner sep=0pt, outer sep=0pt] at (asb) {};
\node[circle, thick, draw, fill=gray, minimum size=6, inner sep=0pt, outer sep=0pt] at (asc) {};
\node[circle, thick, draw, fill=gray, minimum size=6, inner sep=0pt, outer sep=0pt] at (asd) {};

\node[circle, thick, draw, fill=gray, minimum size=6, inner sep=0pt, outer sep=0pt] at (bca) {};
\node[circle, thick, draw, fill=gray, minimum size=6, inner sep=0pt, outer sep=0pt] at (bcb) {};
\node[circle, thick, draw, fill=gray, minimum size=6, inner sep=0pt, outer sep=0pt] at (bcc) {};
\node[circle, thick, draw, fill=gray, minimum size=6, inner sep=0pt, outer sep=0pt] at (bsa) {};
\node[circle, thick, draw, fill=gray, minimum size=6, inner sep=0pt, outer sep=0pt] at (bsb) {};

\node[circle, thick, draw, fill=gray, minimum size=6, inner sep=0pt, outer sep=0pt] at (cca) {};
\node[circle, thick, draw, fill=gray, minimum size=6, inner sep=0pt, outer sep=0pt] at (ccb) {};
\node[circle, thick, draw, fill=gray, minimum size=6, inner sep=0pt, outer sep=0pt] at (ccc) {};
\node[circle, thick, draw, fill=gray, minimum size=6, inner sep=0pt, outer sep=0pt] at (csa) {};
\node[circle, thick, draw, fill=gray, minimum size=6, inner sep=0pt, outer sep=0pt] at (csb) {};
\node[circle, thick, draw, fill=gray, minimum size=6, inner sep=0pt, outer sep=0pt] at (csc) {};
\node[circle, thick, draw, fill=gray, minimum size=6, inner sep=0pt, outer sep=0pt] at (csd) {};

\draw[decoration={brace}, decorate, thick] (0.825, 2.4) -- node [above, pos=0.5] {$\Phi$} (3.675, 2.4);
\draw[decoration={brace}, decorate, thick] (6.075, 2.4) -- node [above, pos=0.5] {$\Phi$} (7.425, 2.4);
\draw[decoration={brace, mirror}, decorate, thick] (9.075, -0.4) -- node [below, pos=0.5] {$\Phi$} (11.175, -0.4);

\node (t1) at (2.25, -0.6) {$\alpha = 1.75$};
\node (t2) at (6.75, -0.6) {$\alpha = 1.5$};
\node (t3) at (10.125, 2.6) {$\alpha = 0.75$};

\node at (-1, 0) {$Q$};
\node at (-1, 2) {$P$};
\end{tikzpicture}}
\caption{A visual representation of the block decomposition of $H$ as trapezoids together with edges of $M^\star \setminus H$. If $e$ turns out to be one of the solid edges it {\it must} be taken into $\MatchingSkeleton(G_1)$; however, if it is one of the dotted edges, it might not be.}
\label{figure:mainproof}
\end{center}
\vskip -0.2in
\end{figure}
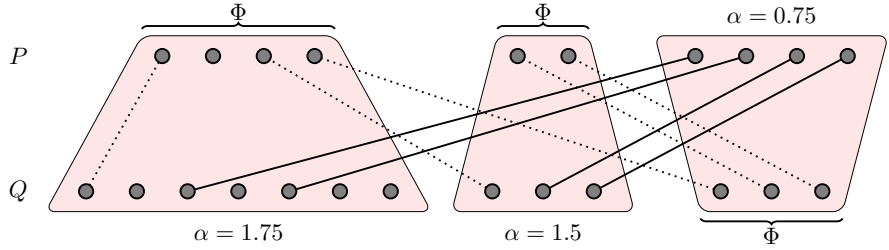

\section{Limitations of the matching skeleton	 coreset}\label{section:limitations}
In this section we show the limits of any matching skeleton as a randomized composable coreset for maximum matching by constructing a pathological bipartite graph on which it only preserves the maximum matching size to a factor of $2/3$.

\begin{theorem}
\label{theorem:main-upper}
For large enough $n$ and $k$ such that $k=O(n/\log n)$, $k=\omega(1)$, there exists a bipartite graph $G$ on $n$ vertices with maximum matching size $\mu$, for which the maximum matching size of
$$\widetilde G=\bigcup_{i=1}^k\MatchingSkeleton(G_i)$$
is at most $\mu\cdot(2/3+o(1))$ with high probability.
\end{theorem}

\begin{remark}
Note that here the high probability is over the randomness of the partition. The choice of the matching skeleton is considered to be adversarial in each subgraph, among the multiple possible valid choices.
\end{remark}

We begin by defining the graph $\bigraph$. The construction follows the ideas used in\cite{DBLP:conf/soda/AssadiBBMS19} to prove an upper bound on the performance of the maximum matching coreset. Let the vertex-set of $G$ consist of six parts: $P_1$, $P_2$, and $P_3$ make up $P$ on one side of the bipartition and $Q_1$, $Q_2$, and $Q_3$ make up $Q$ on the other side. Let the sizes of $P_1$, $P_2$, $Q_2$, and $Q_3$ be $r$, and let the sizes of $Q_1$ and $P_3$ be $r+2r/k$, where $r$ is some parameter such that $6r+4r/k=n$. The edge-set $E$ is comprised of the following:
\begin{itemize}
\item A perfect matching between all of $P_1$ and a subset of $Q_1$,
\item a complete bipartite graph between $Q_1$ and $P_2$,
\item a perfect matching between $P_2$ and $Q_2$,
\item a complete bipartite graph between $Q_2$ and $P_3$,
\item and a perfect matching between a subset of $P_3$ and all of $Q_3$.
\end{itemize}

The graph is pictured in Figure \ref{figure:zzgraph}. The analysis of the behavior of $\MatchingSkeleton$ on this graph relies on the observation that in a typical subsampled version, $P_1\cup Q_1\cup P_2$ forms a region of $\alpha$-value at least $1$ while $Q_2\cup P_3\cup Q_3$ forms a region of $\alpha$-value at most $1$. This means that the edges sampled between $P_2$ and $Q_2$ need not be taken into the matching skeleton, which further implies that $\widetilde G$ can be missing the entire $(P_2,Q_2)$ matching.\\

In order to prove this we will need the following basic property of expansion levels. One side of the lemma has been previously shown in \cite{DBLP:conf/soda/BernsteinHR18}.

\begin{lemma}
\label{lemma:upper-lemma}
Consider a bipartite graph $\bigraph$.
\begin{itemize}
\item If $P$ can be perfectly matched to $Q$, then $\min\alpha\ge1$.
\item Conversely, if $Q$ can be perfectly matched to $P$, then $\max\alpha\le1$.
\end{itemize}
\end{lemma}
\begin{proof}

By optimality of the canonical vertex cover (Lemma~\ref{lemma:canonicalvc}), and by Theorem~\ref{thm:vc=mm} we have that the size of the maximum matching is
$$\sum_{i=1}^k\begin{cases}|P_i|\text{ if $\alpha_i\ge1$}\\|Q_i|\text{ if $\alpha_i < 1$}\end{cases}.$$
In the first case of the lemma, $\mm{G}=|P|=\sum_{i=1}^k|P_i|$, therefore $\alpha_i$ must always be at least $1$. In the second case, $\mm{G}=|Q|=\sum_{i=1}^k|Q_i|$, therefore $\alpha_i$ must always be at most $1$.
\end{proof}

Finally, we state a result on perfect matchings in random bipartite graphs. This is a simplification, and direct result of Corollary 7.13 from \cite{PERFECT_MATCHING_IN_RANDOM}.

\begin{theorem}
\label{theorem:BB}
Let $H$ be a random bipartite graph on $n+n$ vertices, where each of the $n^2$ possible edges appears independently with probability $p=\Omega(\log n/n)$. Then $H$ contains a perfect matching with high probability.
\end{theorem}

We are ready to prove Theorem~\ref{theorem:main-upper}.

\begin{proofof}{Theorem~\ref{theorem:main-upper}}

Consider $G_i=(P,\, Q,\, E_i)$, the graph $G$ sub-sampled at rate $1/k$. We claim that with high probability the non-isolated vertices of $P_1\cup P_2$ can be perfectly matched to $Q_1$. Indeed, out of $r$ edges of $P_1\times Q_1$,  we expect $r/k$ of them to  appear in $G_i$ and with high probability no more than $2r/k$ do (see Theorem~\ref{theorem:chernoff}). In this case, at least $r$  unmatched vertices of $Q_1$ remain, which we will call $Q_1'$. Note that the graph between $Q_1'$ and $P_2$ follows the same distribution as the random graph described in Theorem~\ref{theorem:BB}, with $p=1/k=\Omega(\log n/n)$. Therefore, $(Q_1'\times P_2)\cap E_i$ contains a perfect matching with high probability.\\

By Lemma~\ref{lemma:upper-lemma}, this means that the subgraph induced by $P_1\cup Q_1\cup P_2$ in $G_i$ has block decomposition with all $\alpha\ge1$. By similar reasoning we can show that the non-isolated vertices of $Q_2\cup Q_3$ can be perfectly matched to $P_3$. Hence, by Lemma~\ref{lemma:upper-lemma}, the induced subgraph of $Q_2\cup P_3\cup Q_3$ in $G_i$ has a block decomposition with all $\alpha\le1$.\\

Simply taking the disjoint union of these two induced subgraphs does not change the expansion levels. Hence the graph $G_i^-$, consisting of all edges of $G_i$ {\it except those between $P_2$ and $Q_2$}, has block decomposition with the $\alpha$ values of $P_1$, $Q_1$, and $P_2$ being at least $1$, and the $\alpha$ values of $Q_2$, $P_3$ and $Q_3$ being at most $1$. Let $H$ be a matching skeleton of $G_i^-$. By applying Lemma~\ref{lemma:robust} with $E^-=\emptyset$ and $E^+=E_i\cap P_2\times Q_2$, we get that $H$ is still a matching skeleton of $G_i$. Therefore, there exists a matching skeleton of $G_i$ which contains no edges from $P_2\times Q_2$.\\

In conclusion, it is possible that each coreset selects a matching skeleton of its sub-graph containing no edges from $Q_2 \times P_2$. In such case, the maximum matching of $\widetilde{G}$ has size at most $2r + 4r/k$, whereas that of $G$ was $3r$.
\end{proofof}

\begin{remark}
With a simple alteration to the proof, it can be shown that this upper bound holds even when the individual matching skeletons are selected arbitrarily.
\end{remark}

\begin{figure}[ht]
\vskip 0.2in
\begin{center}
\scalebox{0.9}{
\begin{tikzpicture}
\path[draw, thick] (-0.30, 0) -- (-0.30, -4);
\path[draw, thick] (-0.15, 0) -- (-0.15, -4);
\path[draw, thick] (0, 0) -- (0, -4);
\path[draw, thick] (0.15, 0) -- (0.15, -4);
\path[draw, thick] (0.30, 0) -- (0.30, -4);

\path[draw, thick] (4.70, 0) -- (4.7, -4);
\path[draw, thick] (4.85, 0) -- (4.85, -4);
\path[draw, thick] (5, 0) -- (5, -4);
\path[draw, thick] (5.15, 0) -- (5.15, -4);
\path[draw, thick] (5.30, 0) -- (5.30, -4);

\path[draw, thick] (9.70, 0) -- (9.7, -4);
\path[draw, thick] (9.85, 0) -- (9.85, -4);
\path[draw, thick] (10, 0) -- (10, -4);
\path[draw, thick] (10.15, 0) -- (10.15, -4);
\path[draw, thick] (10.30, 0) -- (10.30, -4);

\draw[fill=gray, thick, rotate around={-40:(2.5, -2)}] (-1,-1.6) rectangle node[rotate=-40] {complete bipartite} (6,-2.4);

\draw[fill=gray, thick, rotate around={-40:(7.5, -2)}] (4, -1.6) rectangle node[rotate=-40] {complete bipartite} (11,-2.4);

\path[draw, dashed] (-2, -1.4) -- (0.5, -1.4);
\path[draw, dashed] (-2, -2.6) -- (0.5, -2.6);
\node[rectangle, align=center, minimum width=2, minimum height=1] at (-1.3, -2) {matching\\ of size $n$};

\node[cloud, cloud puffs=13, draw, thick, fill=data_color, minimum width=2cm, minimum height=2cm, align=center] (t1) at (0, 0) {$P_3$ \\ $r+ \frac{2r}{k}$};
\node[cloud, cloud puffs=13, draw, thick, fill=data_color, minimum width=2cm, minimum height=2cm, align=center] (t2) at (5, 0) {$P_2$ \\ $r$};
\node[cloud, cloud puffs=13, draw, thick, fill=data_color, minimum width=2cm, minimum height=2cm, align=center] (t3) at (10, 0) {$P_1$ \\ $r$};

\node[cloud, cloud puffs=13, draw, thick, fill=data_color, minimum width=2cm, minimum height=2cm, align=center] (b1) at (0, -4) {$Q_3$ \\ $r$};
\node[cloud, cloud puffs=13, draw, thick, fill=data_color, minimum width=2cm, minimum height=2cm, align=center] (b2) at (5, -4) {$Q_2$ \\ $r$};
\node[cloud, cloud puffs=13, draw, thick, fill=data_color, minimum width=2cm, minimum height=2cm, align=center] (b3) at (10, -4) {$Q_1$ \\ $r+ \frac{2r}{k}$};
\end{tikzpicture}
}
\caption{The pathological graph. A typical sub-sampling of the graph has a matching skeleton that does not contain any edges of $Q_2 \times P_2$.}
\label{figure:zzgraph}
\end{center}
\vskip -0.2in
\end{figure}
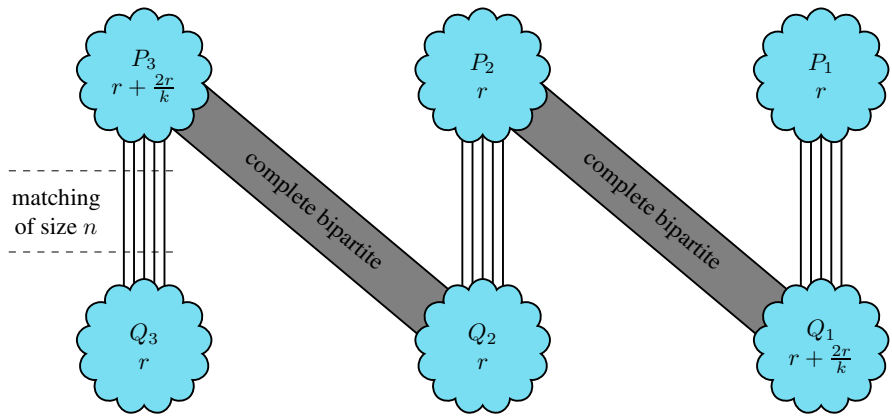

\section*{Acknowledgement}
This project has received funding from the European Research Council (ERC) under the European Union's Horizon 2020 research and innovation programme (grant agreement No 759471)

\newcommand{\etalchar}[1]{$^{#1}$}

\end{document}